\documentclass[twocolumn,amsthm]{autart}

\usepackage{ifpdf}
\newif\ifarxivhtml
\ifpdf
  \arxivhtmlfalse
\else
  \arxivhtmltrue
\fi
\usepackage[T1]{fontenc}
\usepackage[utf8]{inputenc}
\usepackage{lmodern}
\usepackage{amsmath,amssymb,mathtools}
\usepackage{booktabs}
\usepackage{tabularx}
\usepackage{graphicx}
\usepackage{tikz}
\usetikzlibrary{calc,arrows.meta,positioning}
\graphicspath{{figures/}{./}}
\usepackage{enumitem}
\usepackage{url}
\usepackage[authoryear,round]{natbib}

\theoremstyle{plain}
\newtheorem{theorem}{Theorem}
\newtheorem{lemma}[theorem]{Lemma}
\newtheorem{corollary}[theorem]{Corollary}
\newtheorem{proposition}[theorem]{Proposition}

\theoremstyle{remark}
\newtheorem{remark}[theorem]{Remark}

\newcommand{\diag}{\operatorname{diag}}
\newcommand{\range}{\operatorname{range}}
\DeclareMathOperator{\adj}{adj}

\DeclareMathOperator{\disc}{disc}
\newcommand{\one}{\mathbf{1}}

\newsavebox{\bioPhoto}
\newlength{\bioTopT}
\newcommand{\biographyentry}[3]{%
  \par\bigskip
  \sbox{\bioPhoto}{%
    \begin{minipage}[b][1.25in][c]{1.0in}%
      \centering
      \includegraphics[width=1.0in,height=1.25in,keepaspectratio]{#1}%
    \end{minipage}%
  }%
  \settoheight{\bioTopT}{\mbox{T}}%
  \hangindent=1.14in
  \hangafter=-8
  \noindent\makebox[0pt][l]{%
    \hspace{-\hangindent}%
    \raisebox{\bioTopT}[0pt][0pt]{%
      \raisebox{-\ht\bioPhoto}[0pt][0pt]{\usebox{\bioPhoto}}%
    }%
  }%
  \noindent\textbf{#2}\ #3%
  \par\bigskip}

\ifarxivhtml
  \renewcommand{\vec}{\operatorname{vec}}
\fi
\newcommand{\LyapunovCoreNotationTable}{%
\begin{table*}[t]
\centering\small
\setlength{\tabcolsep}{5pt}
\renewcommand{\arraystretch}{1.04}
\caption{Core notation; other symbols are defined at first use.}
\label{tab:notation-core}
\begin{tabularx}{\textwidth}{@{}>{\raggedright\arraybackslash}p{.29\textwidth} >{\raggedright\arraybackslash}X@{}}
\toprule
Symbol & Meaning \\
\midrule
\(A_2=A+vw^\top;\quad \sigma(M)\) & Rank-one state-matrix update; eigenvalue multiset of \(M\), counted with algebraic multiplicity. \\
\(P,Q,R\) & Unique solutions of the two Lyapunov equations and the cross-Sylvester equation \((3)\)--\((5)\). \\
\(\mathcal U\subset\Omega^\circ\subset\Omega\subset\mathcal D\) & Modal seed; disjoint-spectrum domain; main-identity domain; largest domain on which \(P,Q,R\) are uniquely defined. \\
\(\chi_{A_2},\chi_{A_2}^*;\quad g(z),b(z),B\) & Characteristic and reciprocal polynomials; resolvent vector, reciprocal quotient, and the modal representation of \(b(A)\). \\
\(\delta_P,\delta_Q,\delta_R;\quad \widehat P,\widehat Q,\widehat R,\widehat N\) & Lyapunov-operator determinants and the polynomial numerator matrices obtained after clearing them. \\
\(Z=b(A)^{-1}P\) & Explicit similarity matrix on the disjoint-spectrum domain \(\Omega^\circ\). \\
\(\mathcal C_\infty,\mathcal O_\infty;\quad \theta_i,r_i\) & Reachability and inverse-model observability maps; their principal angles and the past/future process correlations. \\
\(\mathcal E_\pm,\mathcal Y_\pm;\quad H,\mathcal X\) & Past/future innovation and output spaces; finite-rank Hankel map and its active past state subspace. \\
\bottomrule
\end{tabularx}
\par\vspace{.2em}\footnotesize\raggedright Here \(M^\top\) is the algebraic (non-conjugating) transpose and \(\chi_{A_2}^*\) is the reciprocal polynomial, not an adjoint.
\end{table*}
}

\date{25 August 2026}

\begin{document}

\begin{frontmatter}

\runtitle{Kernel proof of the De Cock--De Moor identity}

\title{A kernel proof of the De Cock--De Moor
Lyapunov identity}

\ifarxivhtml
\author{Jonas Gillberg\textsuperscript{a,*}}
\else
\author[stockholm]{Jonas Gillberg\thanksref{footnoteinfo}}
\ead{jonas@zyqe.se}
\address[stockholm]{ZyQE, Stockholm, Sweden}
\fi
\ifarxivhtml
\author{Johan L\"ofberg\textsuperscript{b}}
\else
\author[linkoping]{Johan L\"ofberg}
\ead{johan.lofberg@liu.se}
\address[linkoping]{Division of Automatic Control, Department of
Electrical Engineering, Link\"oping University, Link\"oping, Sweden}
\fi

\ifarxivhtml
\else
\thanks[footnoteinfo]{Corresponding author: J.~Gillberg
(\texttt{jonas@zyqe.se}).}
\fi


\begin{abstract}
We prove the rank-one Lyapunov spectral identity recorded as Problem~9.1 in
the 2004 collection of unsolved problems in mathematical systems and control
theory.  Let $P,Q,R$ solve the coupled discrete Lyapunov and Sylvester
equations associated with $A$ and its rank-one update $A_2=A+vw^\top$.  When
the displayed inverses exist, we show that
$P^{-1}RQ^{-1}R^\top$ and $(I+PQ)^{-1}$ have the same characteristic
polynomial.  A rank-one determinant factorization of the equation for $Q$
produces a scalar bilinear kernel.  Evaluating it at the eigenvalues of $A$
and at their reciprocals gives
$RQ^{-1}R^\top=BQ^{-1}B=P-BPB$, after which the two target matrices are the
same two factors in opposite order.  Polynomial continuation extends
the identity from a nonempty open set of admissible systems to the full
admissible domain and
yields a determinant corollary without stability assumptions; when the spectra
of $A$ and $A_2$ are disjoint, $Z=b(A)^{-1}P$ gives an explicit similarity.
In the Schur-stable realization setting, the result recovers the associated
principal-angle and past/future canonical-correlation spectra.
\end{abstract}

\ifarxivhtml
\par\smallskip
\noindent\textsuperscript{a}ZyQE, Stockholm, Sweden.\par
\noindent\textsuperscript{b}Division of Automatic Control, Department of Electrical Engineering, Link\"oping University, Link\"oping, Sweden.\par
\smallskip
\noindent\textsuperscript{*}\textit{Corresponding author: J.~Gillberg (\texttt{jonas@zyqe.se}).}
\fi

\end{frontmatter}

\section{Introduction}\label{sec:intro}

How much does the past of a process reveal about its future? In stochastic
realization theory the answer is geometric: the canonical correlations between
past and future are the sines of the principal angles between a controllability
and an observability subspace. Following this to its algebraic core, the geometry reduces to a question about
three matrices. For
$A\in\mathbb R^{n\times n}$ and $v,w\in\mathbb R^n$, form the rank-one update
$A_2=A+vw^\top$ and let $P,Q,R$ solve the discrete Lyapunov and Sylvester equations
they generate; the question is whether
\begin{equation}\label{eq:intro-iso}
  P^{-1}RQ^{-1}R^\top \qquad\text{and}\qquad (I+PQ)^{-1}
\end{equation}
always share the same eigenvalues.

The proof has two algebraic steps. We first establish the
characteristic-polynomial identity on a nonempty open set of admissible data,
where the target matrices are products of the same two factors in reverse
order. After clearing denominators, every coefficient of their
characteristic-polynomial difference is a polynomial in the entries of
$A,v,w$. Since these polynomials vanish on a nonempty open set, they vanish
identically, and division by the nonzero denominators gives the result
throughout the natural domain. This is precisely Problem~9.1 in the 2004
collection of open problems in systems and control
\citep{BlondelMegretski2004}.

The identity is due to De~Cock and De~Moor, who checked it numerically, proved it
for $n=1$, and established the underlying relation by realization-theoretic means
\citep{DeCock2002,DeCockDeMoor2002}; \citet{Scherrer2002} treated the associated
relation by minimum-phase balanced truncation.

Building on that realization-theoretic setting, the present paper makes three
complementary contributions.  First, it gives a direct finite-dimensional
proof in which a Lyapunov-kernel identity reduces the spectral statement to
the classical $AB$--$BA$ principle.  Second, it identifies an explicit
similarity transformation when the spectra of $A$ and $A_2$ are disjoint.
Third, after clearing denominators, it extends the
characteristic-polynomial identity to the full domain excluding reciprocal
eigenvalue products and obtains the determinant corollary without stability
or definiteness assumptions.

On the open set, a coordinate change makes
$A=\diag(\alpha_1,\ldots,\alpha_n)$; $A_2$ need not be diagonalized.  For the
Lyapunov equation $Q=A_2^\top QA_2+ww^\top$, set
\[
 g(z)=(I-zA_2^\top)^{-1}w,\qquad
 b(z)=\frac{\det(zI-A_2)}{\det(I-zA_2)}.
\]
A rank-one determinant factorization gives
\[
 g(z)^\top Q^{-1}g(\zeta)=\frac{1-b(z)b(\zeta)}{1-z\zeta},
 \qquad b(z)b(1/z)=1.
\]
Evaluating at $z=\alpha_i$ and at $z=1/\alpha_i$ gives the single matrix
relation
\begin{align}
 RQ^{-1}R^\top&=BQ^{-1}B=P-BPB, \label{eq:intro-kernel}\\
 B&=\diag(b(\alpha_1),\ldots,b(\alpha_n)). \notag
\end{align}
Product reversal then proves equality of the characteristic polynomials on
\(\mathcal U\).
Thus evaluation at the eigenvalues and at their reciprocals gives two
instances of the same
Lyapunov kernel, and the spectral step is the elementary reversal of two
matrix factors.

The identity becomes polynomial after denominators are cleared and therefore
needs no stability or definiteness hypothesis.  The Schur-stable realization
in which it first arose is recovered separately in Proposition~\ref{prop:p91}.
Sections~\ref{sec:setup}--\ref{sec:proof} give the matrix proof and its
continuation; Sections~\ref{sec:intertwiner}--\ref{sec:consequences} give the
explicit similarity and the system-theoretic consequences.

\section{Statement}\label{sec:statement}

We first state the matrix identity and then its realization-theoretic
consequence.

Let $A \in \mathbb R^{n\times n}$ ($n\ge1$) and $v, w \in \mathbb R^n$, and set
$A_2 = A + vw^\top$. Throughout, $\sigma(\cdot)$ denotes the multiset of
complex eigenvalues counted with algebraic multiplicity, and
$\alpha_1,\dots,\alpha_n$, $\beta_1,\dots,\beta_n$ enumerate $\sigma(A)$
and $\sigma(A_2)$. An indexed family in braces, such as
$\{\cos^2\theta_i\}_{i=1}^n$, is likewise read as a multiset, with repetitions
retained. Assume that no product of two eigenvalues drawn from
$\sigma(A)\cup\sigma(A_2)$ equals $1$ --- self-pairs and repeats included:
$\alpha_i\alpha_j\neq1$, $\beta_i\beta_j\neq1$, and $\alpha_i\beta_j\neq1$
for all $i,j$ --- equivalently, that
$\delta_P:=\det(I-A\otimes A)$,
$\delta_Q:=\det(I-A_2^\top\otimes A_2^\top)$,
$\delta_R:=\det(I-A_2^\top\otimes A)$ are all nonzero. This is
exactly the condition under which the three coupled equations
\begin{align}
  P &= APA^\top + vv^\top, \label{eq:P}\\
  Q &= A_2^\top Q A_2 + ww^\top, \label{eq:Q}\\
  R &= ARA_2 + vw^\top \label{eq:R}
\end{align}
have unique solutions $P, Q, R \in \mathbb R^{n\times n}$ --- \eqref{eq:P} and
\eqref{eq:Q} are discrete Lyapunov equations, \eqref{eq:R} a discrete Sylvester
equation --- with $P, Q$ symmetric (\citealp[\S12.5]{LancasterRodman1995};
\citealp[\S11.21 and~\S12.21]{Bernstein2009}).
\begin{remark}[Conventions: $r_i$]\label{rem:conventions}
The past/future canonical correlations are written
$r_1\ge r_2\ge\cdots$, while the principal angles are ordered as
$\theta_1\le\dots\le\theta_n$. In the system-theoretic consequence of
Problem~9.1, the largest $n$ correlations satisfy
$r_i=\sin\theta_{n+1-i}$, $1\le i\le n$, and all remaining correlations
vanish: $r_{n+1}=r_{n+2}=\cdots=0$. The $\rho_i$ of
\citet{BlondelMegretski2004} are our $r_i$.
\end{remark}

\begin{remark}[Relation to Problem~9.1]\label{rem:p91verbatim}
Equations \eqref{eq:P}--\eqref{eq:R} are exactly the three independent blocks
of the problem's single $2n\times2n$ Lyapunov equation
\citep[p.~288]{BlondelMegretski2004},
\[
  \Pi=\mathcal A\,\Pi\,\mathcal A^\top
      +\begin{pmatrix}v\\w\end{pmatrix}\begin{pmatrix}v^\top&w^\top\end{pmatrix},
\]
where $\Pi:=\bigl(\begin{smallmatrix}P&R\\R^\top&Q\end{smallmatrix}\bigr)$,
$\mathcal A:=\bigl(\begin{smallmatrix}A&0\\0&A_2^\top\end{smallmatrix}\bigr)$,
and $A_2=A+vw^\top$: its $(1,1)$, $(2,2)$, and $(1,2)$ blocks are
\eqref{eq:P}, \eqref{eq:Q}, \eqref{eq:R} (the $(2,1)$ block is the transpose
of \eqref{eq:R}). The original hypothesis --- no two eigenvalues
$\lambda_i,\lambda_j$ of $\bigl(\begin{smallmatrix}A&0\\0&A_2\end{smallmatrix}\bigr)$
with $\lambda_i\lambda_j=1$, $i,j=1,\dots,2n$ --- is precisely
$\delta_P\delta_Q\delta_R\neq0$, which (loc.\ cit., and as above) makes the
block equation uniquely solvable. The stated conclusion --- that, when
$P$, $Q$, $I+PQ$ are nonsingular, ``$P^{-1}RQ^{-1}R^\top$ and $(I_n+PQ)^{-1}$
have the same eigenvalues'' --- is verbatim the assertion of
Theorem~\ref{thm:main}, with $\sigma(\cdot)$ the eigenvalue multiset. No
stability, definiteness, or stochastic interpretation is assumed at this level;
those enter only in the realization reading of Proposition~\ref{prop:p91}.
The rank-one hypothesis is essential for the unrestricted assertion: the
original problem explicitly records counterexamples when $v,w$ are replaced by
$V,W\in\mathbb R^{n\times m}$, $m>1$ \citep[p.~288]{BlondelMegretski2004}. This
records the existence of higher-rank counterexamples; it does not assert
failure for every higher-rank datum.

The principal-angle and canonical-correlation interpretation of this matrix
identity is stated explicitly in Proposition~\ref{prop:p91}.
\end{remark}

\begin{theorem}[\textnormal{Rank-one Lyapunov spectral identity, proved in
Section~\ref{sec:proof}}]\label{thm:main}%
Let $(A,v,w)\in\mathcal D$ --- $A\in\mathbb R^{n\times n}$, $v,w\in\mathbb R^n$,
$A_2=A+vw^\top$, no two eigenvalues among $\sigma(A)\cup\sigma(A_2)$ with
product $1$, self-pairs and repeats included, i.e.\
$\delta_P\delta_Q\delta_R\neq0$ --- so that
\eqref{eq:P}--\eqref{eq:R} have unique solutions $P,Q,R$.
Assume in addition that $P$, $Q$, and $I+PQ$ are nonsingular.
Then
\[
  \sigma\bigl(P^{-1}RQ^{-1}R^\top\bigr) = \sigma\bigl((I+PQ)^{-1}\bigr).
\]
\end{theorem}

\begin{proposition}[System-theoretic consequence of Problem~9.1]\label{prop:p91}
Consider the real scalar forward-innovation model driven by the scalar,
zero-mean, unit-variance white innovation sequence $\{\varepsilon(k)\}$,
\[
  x(k+1)=Ax(k)+K\varepsilon(k),\qquad
  y(k)=c^\top x(k)+\varepsilon(k),
\]
and make the substitution used in the original problem,
\[
  v=K,\qquad w=-c,\qquad A_2=A+vw^\top=A-Kc^\top.
\]
Assume that $A$ and $A_2$ are Schur-stable, $(A,v)$ is controllable, and
$(A_2,w)$ is observable.
With
\[
  \begin{aligned}
    \mathcal C_\infty&=[v,Av,A^2v,\dots],\\
    \mathcal O_\infty&=[w,A_2^\top w,(A_2^\top)^2w,\dots],
  \end{aligned}
\]
the solutions of \eqref{eq:P}--\eqref{eq:R} are
\[
  P=\mathcal C_\infty\mathcal C_\infty^\top,\qquad
  Q=\mathcal O_\infty\mathcal O_\infty^\top,\qquad
  R=\mathcal C_\infty\mathcal O_\infty^\top.
\]
If $\theta_1\le\dots\le\theta_n$ are the principal angles between
$\range(\mathcal C_\infty^\top)$ and $\range(\mathcal O_\infty^\top)$, then
\[
  \{\cos^2\theta_i\}_{i=1}^{n}
    =\sigma(P^{-1}RQ^{-1}R^\top).
\]
The original conjecture in Problem~9.1 is the pure matrix identity of
Theorem~\ref{thm:main}; it does not require this stochastic realization. In
the forward-innovation realization above, that theorem gives
\[
  \{\cos^2\theta_i\}_{i=1}^{n}
    =\sigma((I+PQ)^{-1}).
\]
Let $r_1\ge r_2\ge\cdots\ge0$ be the canonical correlations between the
closed linear spans of the past and future outputs.
Then
\[
  \begin{aligned}
    r_i&=\sin\theta_{n+1-i} &&(1\le i\le n),\\
    r_{n+1}&=r_{n+2}=\cdots=0,
  \end{aligned}
\]
or, equivalently,
$\{r_i^2\}_{i=1}^{n}=\sigma\bigl(I-(I+PQ)^{-1}\bigr)$.
Off the realization face no stability or stochastic hypothesis is used: the
cleared identity \eqref{eq:cleared} holds on all of $\mathcal D$
(Lemma~\ref{lem:divideback}), and the spectral identity holds wherever
$P,Q,I+PQ$ are nonsingular.
\end{proposition}

\begin{proof}
The proof uses three standard reductions from systems theory.  First,
normalized reachability and observability maps identify the principal-angle
spectrum.  Second, stable causal inversion represents the future-output space
as the graph of a finite-rank Hankel map.  Third, least-squares projection onto
that graph reduces the infinite-dimensional prediction problem to one
$n$-dimensional normal equation.  The five steps below give the details.

\emph{1. Stable realization and Gramians.}
Schur stability makes the controllability, observability, and cross-Gramian
series absolutely convergent. Shifting their summation indices gives
\[
\begin{aligned}
  P&=APA^\top+vv^\top,\\
  Q&=A_2^\top Q A_2+ww^\top,\\
  R&=ARA_2+vw^\top.
\end{aligned}
\]
Uniqueness therefore identifies these series with the solutions of
\eqref{eq:P}--\eqref{eq:R}. Controllability and observability give full row
rank. Hence the normalized reachability and observability maps
\[
  C_{\mathrm n}=P^{-1/2}\mathcal C_\infty,
  \qquad O_{\mathrm n}=Q^{-1/2}\mathcal O_\infty
\]
have orthonormal rows. The cosines of the principal angles are therefore the
singular values of
\[
  C_{\mathrm n}O_{\mathrm n}^\top=P^{-1/2}RQ^{-1/2}
\]
\citep{BjorckGolub1973}.  Squaring gives the spectrum of
$P^{-1/2}RQ^{-1}R^\top P^{-1/2}$, which is similar to
$P^{-1}RQ^{-1}R^\top$.  Theorem~\ref{thm:main} now gives the displayed
$\cos^2\theta_i$ identity.

\emph{2. Innovations and the past prediction space.}
For the process interpretation, take the covariance Hilbert space generated by
the normalized innovations $\{\varepsilon(k)\}_{k\in\mathbb Z}$.  Equivalently,
identify $\varepsilon(k)$ with the $k$th canonical basis vector of
$\mathcal H=\ell^2(\mathbb Z)$.  Under this covariance-preserving
identification, inner products are covariances and orthogonal projections are
optimal linear least-squares predictors.  The relevant closed linear prediction
spaces are
\[
  \begin{aligned}
    \mathcal E_-&:=\overline{\operatorname{span}}\{\varepsilon(k):k<0\},\\
    \mathcal Y_-&:=\overline{\operatorname{span}}\{y(k):k<0\},\\
    \mathcal Y_+&:=\overline{\operatorname{span}}\{y(k):k\ge0\}.
  \end{aligned}
\]
Here $\mathcal E_-$ is the past-innovation space, while $\mathcal Y_-$ and
$\mathcal Y_+$ are the past- and future-output spaces. The innovation-to-output
 transfer has impulse response
 \[
 \begin{aligned}
   y(k)&=\sum_{j\ge0}h(j)\varepsilon(k-j),\\
   h(0)&=1, & h(j+1)&=-w^\top A^j v.
 \end{aligned}
 \]
Schur stability of $A$ makes $h$ absolutely and square summable. The stable
 inverse innovation filter has coefficients
\[
 \eta(0)=1,\qquad \eta(j+1)=w^\top A_2^j v.
\]
Schur stability of $A_2$ gives the same summability for $\eta$. The rank-one
identity $A_2=A+vw^\top$, applied to successive powers, gives the finite
convolution formula
\[
   \sum_{i+j=m}\eta(i)h(j)=\begin{cases}1,&m=0,\\0,&m>0.\end{cases}
\]
Thus the causal filters with coefficients $h$ and $\eta$ are mutual inverses.
Causality and stability then give
\[
  \mathcal Y_-=\mathcal E_-.
\]
Indeed, each past output is generated by past innovations, and the stable inverse
recovers every past innovation from past outputs. Thus past outputs and past
innovations contain exactly the same information for linear prediction.

\emph{3. The future-output space and the Hankel map.}
The inverse filter separates each future output into its new innovation and a
finite-dimensional correction determined by the past state.  This correction
is the Hankel part of the realization.
For $k\ge0$ form the finite, unit-triangular whitening
\[
 y_{\mathrm w}(k):=\sum_{j=0}^{k}\eta(j)y(k-j).
\]
Because this transformation is causal and unit triangular, the closed span of
the $y_{\mathrm w}(k)$ is exactly the future-output space. Splitting the full
inverse convolution at time zero gives
\[
 \begin{aligned}
   y_{\mathrm w}(k)&=\varepsilon(k)
     -\sum_{i=1}^{n}\bigl((A_2^\top)^k w\bigr)_i p_i,\\
   p_i&:=\sum_{\ell\ge0}(A^\ell v)_i\varepsilon(-1-\ell).
 \end{aligned}
\]
Let $\mathcal E_+:=\mathcal E_-^\perp$ and
$q_i:=\sum_{k\ge0}((A_2^\top)^k w)_i\varepsilon(k)$.  The families
$p=(p_i)$ and $q=(q_i)$ lie in $\mathcal E_-$ and $\mathcal E_+$,
respectively, and have Gram matrices $P$ and $Q$. The realization therefore
determines the finite-rank Hankel map from future innovations to the
state-dependent component in the past prediction space:
\[
 H:\mathcal E_+\longrightarrow\mathcal E_-,\qquad
 Hf=\sum_{i=1}^{n}\langle q_i,f\rangle p_i.
\]
Under the orthogonal decomposition
$\mathcal H=\mathcal E_-\oplus\mathcal E_+$, the preceding formula for
$y_{\mathrm w}(k)$
says that the future-output prediction space is the graph of $-H$:
\[
  \mathcal Y_+=\{(-Hf,f):f\in\mathcal E_+\}.
\]

\emph{4. Reduction of the prediction problem to the state space.}
The graph representation is the decisive reduction: projection between two
infinite prediction spaces becomes an ordinary least-squares problem on the
finite state-generated subspace.
Canonical correlations between $\mathcal Y_-=\mathcal E_-$ and
$\mathcal Y_+$ are the cosines of their principal angles. Equivalently, their
squares are the eigenvalues of the past component of the orthogonal projection
onto $\mathcal Y_+$; this projection gives the optimal linear least-squares
predictor between the two spaces. Only the
state-generated part of the past,
$\mathcal X:=\operatorname{span}\{p_1,\ldots,p_n\}$ is active.  For
$u=\sum_i a_i p_i\in\mathcal X$, the least-squares projection of $(u,0)$
onto the graph of $-H$ has the finite-dimensional normal equation
$(I+PQ)d=-Pa$.  Its past component is consequently
\[
  \sum_i\bigl(Q(I+PQ)^{-1}Pa\bigr)_i p_i.
\]
The same projection has zero past component for
$u\in\mathcal E_-\cap\mathcal X^\perp$. Thus the infinite-dimensional
prediction problem reduces exactly to the $n$-dimensional controllable and
observable state contribution.

\emph{5. Canonical-correlation spectrum and ordering.}
The nonzero squared process correlations therefore have characteristic
polynomial
\[
 \chi_{Q(I+PQ)^{-1}P}
   =\chi_{I-(I+PQ)^{-1}},
\]
where the equality follows from product reversal and
$I-(I+PQ)^{-1}=(I+PQ)^{-1}PQ$.  Combining this complement with the already
proved $\cos^2\theta_i$ spectrum, and ordering both nonnegative families in
decreasing order, yields
$r_i=\sin\theta_{n+1-i}$ for $1\le i\le n$.  The vanishing on the inactive
orthogonal complement gives $r_{n+1}=r_{n+2}=\cdots=0$.
\end{proof}

\LyapunovCoreNotationTable

\begin{figure}[t]
\centering
\begin{tikzpicture}[>=stealth,line join=round,line cap=round,scale=1.0]
  \coordinate (O) at (0.35,0.45);
  \coordinate (C) at (4.25,1.55);   
  \coordinate (Op) at (2.85,2.85);  
  \draw[very thick,black] (O) -- (C)
        node[pos=1.0,right=1pt,black] {\footnotesize$\range(\mathcal C_\infty^\top)$};
  \draw[very thick,black,densely dashed] (O) -- (Op)
        node[pos=1.0,above right=-1pt,black] {\footnotesize$\range(\mathcal O_\infty^\top)$};
  \draw[->,thick] (O)++(18:1.05) arc (18:46:1.05);
  \node at ($(O)+(32:1.42)$) {\footnotesize$\theta_i$};
  \fill (O) circle (1.4pt);
  \node[align=center,font=\footnotesize] at (2.35,0.18)
    {$\{\cos^2\theta_i\}_{i=1}^{n}
       \;=\;\sigma\!\bigl(P^{-1}RQ^{-1}R^\top\bigr)
       \;=\;\sigma\!\bigl((I+PQ)^{-1}\bigr)$};
\end{tikzpicture}
\caption{In the Schur-stable controllable/observable case, the theorem says that
the squared cosines of the principal angles $\theta_i$ between the
reachability subspace $\range(\mathcal C_\infty^\top)$ (solid) and the
observability subspace $\range(\mathcal O_\infty^\top)$ (dashed) are the
eigenvalues of $(I+PQ)^{-1}$; the Gramian forms
are in the remark below.}
\label{fig:cc}
\end{figure}
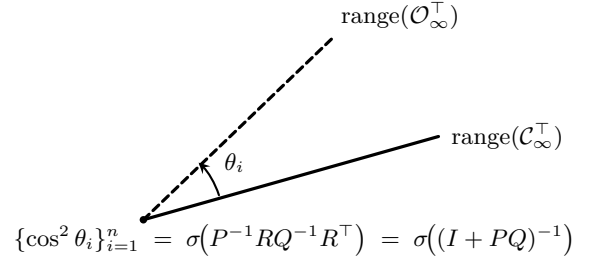

\begin{remark}[Gram matrices in the Schur-stable case]\label{rem:gramians}
When $A,A_2$ are Schur-stable the series solutions of
\eqref{eq:P}--\eqref{eq:R} are Gramians: with the factors
$\mathcal C_\infty$ and $\mathcal O_\infty$ defined in
Proposition~\ref{prop:p91}, one has
$P=\mathcal C_\infty\mathcal C_\infty^\top$,
$Q=\mathcal O_\infty\mathcal O_\infty^\top$, and
$R=\mathcal C_\infty\mathcal O_\infty^\top$, so
\begin{multline*}
P^{-1}RQ^{-1}R^\top
=(\mathcal C_\infty\mathcal C_\infty^\top)^{-1}
 (\mathcal C_\infty\mathcal O_\infty^\top)\\
\mathrel{\phantom{=}}\cdot
 (\mathcal O_\infty\mathcal O_\infty^\top)^{-1}
 (\mathcal O_\infty\mathcal C_\infty^\top)
\end{multline*}
(\citealp[\S4.2]{Antoulas2005}), whose eigenvalues are the squared cosines
$\cos^2\theta_i$ of the principal angles between
$\range(\mathcal C_\infty^\top)$ and $\range(\mathcal O_\infty^\top)$
\citep{BjorckGolub1973,Moore1981} (Figure~\ref{fig:cc}).
Theorem~\ref{thm:main} equates these with $\sigma((I+PQ)^{-1})$, so
$\sigma(PQ)=\{\tan^2\theta_i\}_{i=1}^{n}$.
\end{remark}

\section{Domains, an open set, and change of state coordinates}\label{sec:setup}

The initial Lyapunov-kernel calculation needs only one nonempty
full-dimensional open set in the parameter space.  The stronger polynomial
statement proved later also uses numerator matrices obtained by clearing the
Lyapunov-operator denominators.  We therefore distinguish three domains. The
largest one, on which all three Lyapunov and Sylvester equations have unique
solutions, is
\[
 \mathcal D=\{(A,v,w):\delta_P\delta_Q\delta_R\ne0\}.
\]
The theorem is formulated on
\[
 \Omega=\{(A,v,w)\in\mathcal D:
          \det P\,\det Q\,\det(I+PQ)\ne0\},
\]
and the explicit similarity will be stated on the smaller domain
\[
 \Omega^\circ=\{(A,v,w)\in\Omega:
          \operatorname{res}(\chi_A,\chi_{A_2})\ne0\}.
\]
Thus \(\Omega\) is exactly the subset of \(\mathcal D\) on which the two
matrices in Theorem~\ref{thm:main} are defined, while \(\Omega^\circ\) also
requires \(A\) and \(A_2\) to have disjoint spectra.

\begin{lemma}[A nonempty open set and polynomial uniqueness]\label{lem:seed}%
There is a nonempty Euclidean-open set
\[
             \mathcal U\subset\Omega^\circ
\]
in the original space
\(\mathbb R^{n\times n}\times\mathbb R^n\times\mathbb R^n\).
At
every point of \(\mathcal U\), \(A\) has distinct nonzero eigenvalues and
\((A,v)\) is controllable.
If the coefficients of a polynomial in
\(\lambda\) are rational functions of the entries of \((A,v,w)\), with denominators
nonzero on \(\Omega\), and all those coefficients vanish on \(\mathcal U\),
then they vanish throughout \(\Omega\).
\end{lemma}

\begin{proof}
Take a real diagonal Schur-stable matrix with distinct nonzero eigenvalues and
perturb it by a sufficiently small positive rank-one matrix.  The rank-one
determinant formula gives strict interlacing, while a Vandermonde determinant
gives controllability.  All factors defining \(\Omega^\circ\) are therefore
nonzero at this witness and remain nonzero on a surrounding open ball, which
is the required \(\mathcal U\).  Finally, after a common denominator has been
cleared, each coefficient in the last assertion is a real polynomial.  A real
polynomial that vanishes on the nonempty open set \(\mathcal U\) vanishes
identically.  Appendix~\ref{app:generic-locus} gives the witness and
nonvanishing calculations explicitly.  On the neighborhood there is no simple
spectrum or observability requirement for $A_2$ and no condition on
\(\det A_2\) or \(\det R\).
\end{proof}

Figure~\ref{fig:seed-zero-sets} shows the geometry of these exclusions in the
scalar case and locates the witness used to choose \(\mathcal U\).

\begin{figure}[t]
  \centering
  \includegraphics[width=0.95\linewidth]{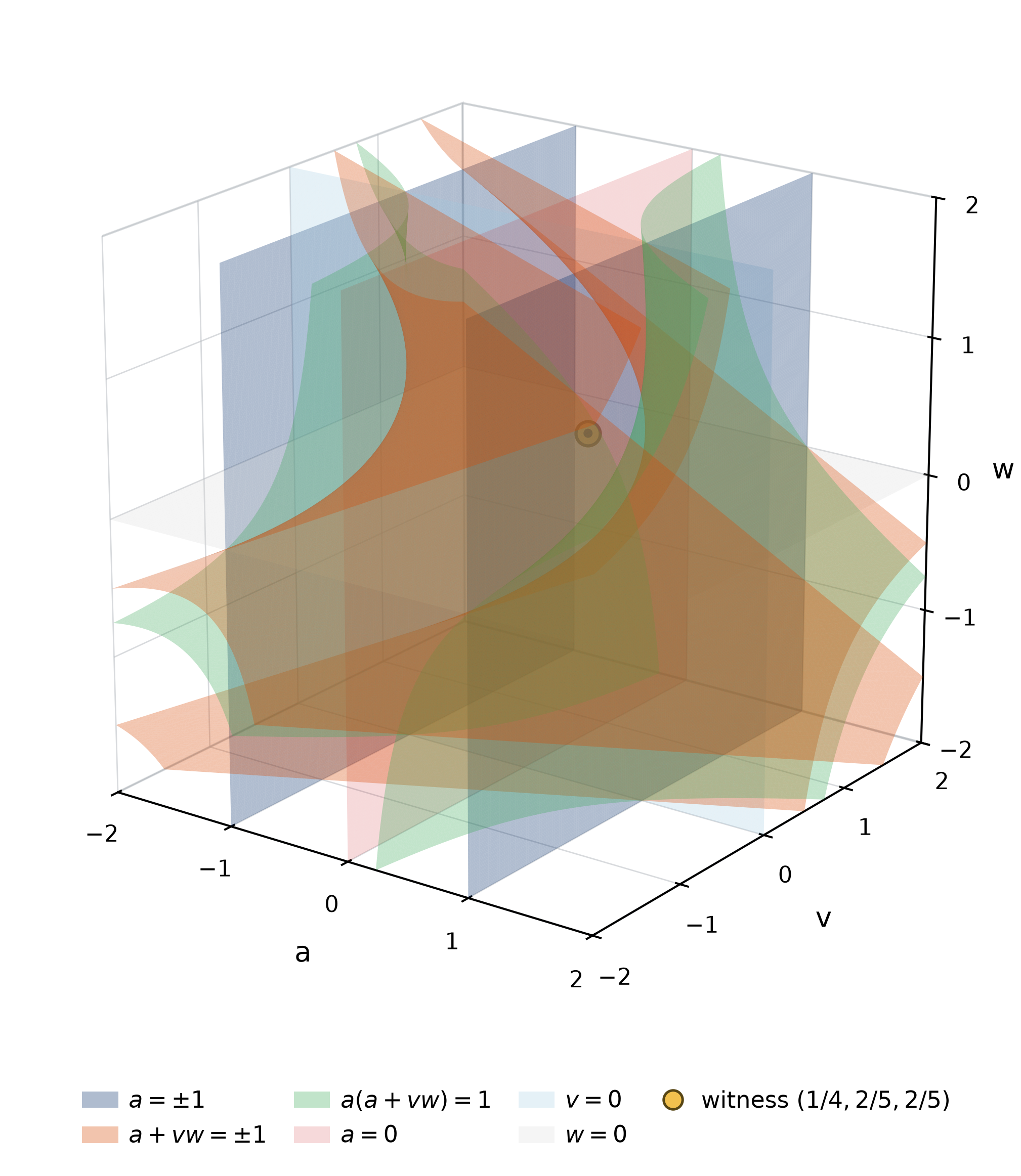}
  \caption{Scalar geometry of the open seed.  For \(n=1\), \(A=a\) and
  \(A_2=a+vw\); the displayed zero sets are
  \(a=\pm1\), \(a+vw=\pm1\), \(a(a+vw)=1\), \(a=0\), \(v=0\), and \(w=0\).
  The cleared numerator of \(\det(I+PQ)\) is \(\delta_R^2\), so it adds no
  further surface.  Any sufficiently small neighborhood of the marked witness
  \((\tfrac14,\tfrac25,\tfrac25)\) that avoids these sets is a valid
  \(\mathcal U\).  The surface \(A_2=0\) is not excluded because the proof does
  not require \(A_2\) to be nonsingular.}
  \label{fig:seed-zero-sets}
\end{figure}

\begin{figure}[t]
\centering
\ifarxivhtml
\includegraphics[
  width=0.95\linewidth,
  alt={Four nested domains: the nonempty open set U is contained in Omega interior, which is contained in Omega, which is contained in D.}
]{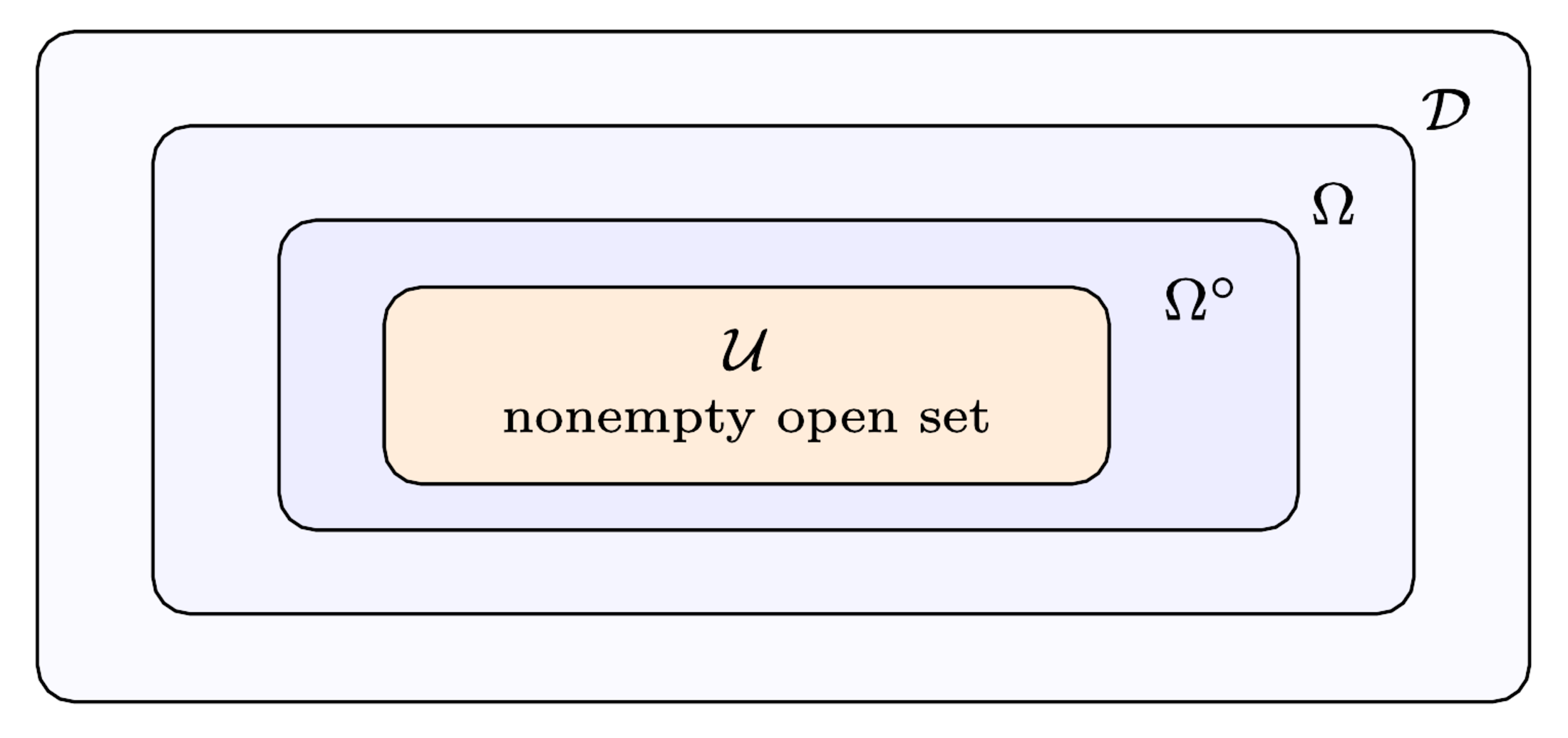}
\else
\begin{tikzpicture}[every node/.style={font=\footnotesize},line width=.45pt]
  \draw[rounded corners=5pt,fill=blue!2] (0,0) rectangle (7.1,3.2);
  \draw[rounded corners=5pt,fill=blue!4] (.55,.42) rectangle (6.55,2.75);
  \draw[rounded corners=5pt,fill=blue!7] (1.15,.82) rectangle (6.0,2.30);
  \draw[rounded corners=5pt,fill=orange!14] (1.65,1.04) rectangle (5.10,1.98);
  \node[anchor=north east] at (6.95,3.05) {$\mathcal D$};
  \node[anchor=north east] at (6.4,2.60) {$\Omega$};
  \node[anchor=north east] at (5.85,2.15) {$\Omega^\circ$};
  \node[align=center] at (3.375,1.51)
    {$\mathcal U$\\[-1pt]{\scriptsize nonempty open set}};
\end{tikzpicture}
\fi
\caption{The four domains have different jobs:
\(\mathcal U\subset\Omega^\circ\subset\Omega\subset\mathcal D\).
The divided kernel calculation starts on \(\mathcal U\), the theorem extends
to \(\Omega\), the explicit similarity extends to \(\Omega^\circ\), and the
denominator-cleared polynomial identity extends to \(\mathcal D\).}
\label{fig:domains}
\end{figure}
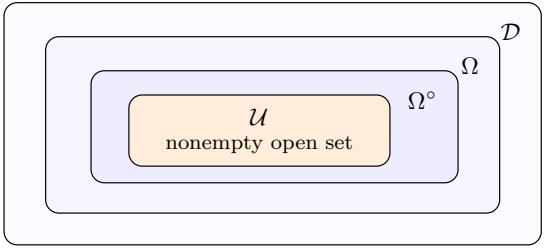

\begin{lemma}[Change to modal coordinates preserves the two spectra]\label{lem:transport}%
Let $(A, v, w) \in \mathcal U$ and let
$T \in \mathbb C^{n\times n}$ be invertible with
$T^{-1} A T = \diag(\alpha)$. Set $\widetilde v := T^{-1} v$,
$\widetilde w := T^\top w$, $\widetilde P := T^{-1} P T^{-\top}$,
$\widetilde Q := T^\top Q T$, $\widetilde R := T^{-1} R T$.
Then
$(\widetilde P, \widetilde Q, \widetilde R)$ solves the Lyapunov
equations~\eqref{eq:P}--\eqref{eq:R} with data $(\diag(\alpha),
\widetilde v, \widetilde w)$.
\[
  \sigma(P^{-1}RQ^{-1}R^\top)
  =\sigma(\widetilde P^{-1}\widetilde R\widetilde Q^{-1}\widetilde R^\top).
\]
\[
  \sigma((I+PQ)^{-1})=\sigma((I+\widetilde P\widetilde Q)^{-1}).
\]
\end{lemma}
\begin{proof}
The Lyapunov equations transform under this change by direct substitution: e.g.\
$P = APA^\top + vv^\top$ multiplied left by $T^{-1}$ and right by
$T^{-\top}$ becomes $\widetilde P=(T^{-1}AT)\widetilde P(T^{-1}AT)^\top+
\widetilde v\widetilde v^\top$, using $T^{-\top}T^\top=I$ and
$T^{-1}v\cdot(T^{-1}v)^\top=\widetilde v\widetilde v^\top$. The other two
Lyapunov equations transform in the same way. For the spectra,
$\widetilde P^{-1}\widetilde R\widetilde Q^{-1}\widetilde R^\top
= T^\top(P^{-1}RQ^{-1}R^\top)T^{-\top}$ by direct multiplication
(four cancellations $T^{-\top}T^\top = I$, $TT^{-1} = I$), a
similarity which preserves the spectrum; the second identity
follows from $\widetilde P\widetilde Q=T^{-1}(PQ)T$.
\end{proof}

\emph{Use of complex modal coordinates.} Although the original matrices are
real, the matrix that diagonalizes \(A\) may be complex. We therefore carry
out the modal calculation over \(\mathbb C\), always using the algebraic
(non-conjugating) transpose. After denominators have been cleared, the
resulting identities are polynomial identities with real coefficients.
Restricting them back to real values of \((A,v,w)\) gives the claimed real
statement.

At a fixed point of \(\mathcal U\), Lemma~\ref{lem:transport} allows us to
work in modal coordinates, in which
\[
                     A=\diag(\alpha_1,\ldots,\alpha_n).
\]
Membership in \(\mathcal U\) gives precisely the nonvanishing needed below:
\(\alpha_i\ne0\), \(v_i\ne0\),
\(1-\alpha_i\alpha_j\ne0\), \(1-\alpha_i\beta_l\ne0\),
\(\alpha_i-\beta_l\ne0\), and nonsingular \(P,Q,I+PQ\).
Thus every resolvent and every value \(b(\alpha_i)\) used in the two kernel
evaluations is well defined, while no inverse of \(R\) is required.  The
corresponding denominator inventory is given in
Appendix~\ref{app:generic-locus}.

\section{The Lyapunov-kernel step}\label{sec:kernel}

Both evaluations used below follow from one determinant identity for the
Lyapunov equation itself. Throughout this section all transposes are
algebraic, non-conjugating transposes.

\begin{lemma}[Lyapunov-kernel factorization]\label{lem:factor}%
Let \(S,Q\in\mathbb C^{n\times n}\) and \(w\in\mathbb C^n\) satisfy
\begin{equation}\label{eq:kernel-lyapunov}
 Q=S^\top QS+ww^\top.
\end{equation}
For all scalars \(z,\zeta\),
\begin{multline}\label{eq:factor}
 (I-\zeta S^\top)Q(I-zS)-(1-z\zeta)ww^\top\\
 =(S^\top-zI)Q(S-\zeta I).
\end{multline}
\end{lemma}

\begin{proof}
Substitute \(ww^\top=Q-S^\top QS\) from
\eqref{eq:kernel-lyapunov} and expand. Both sides are
\[
 S^\top QS-\zeta S^\top Q-zQS+z\zeta Q.
\]
No inverse, division, stability, or definiteness is involved.
\end{proof}

\begin{lemma}[Lyapunov-kernel identity]\label{lem:kernel}%
Assume in addition that \(Q\), \(I-zS\), and \(I-\zeta S\) are nonsingular and
\(1-z\zeta\ne0\).
Define
\[
 g(z)=(I-zS^\top)^{-1}w,\qquad
 b(z)=\frac{\det(zI-S)}{\det(I-zS)}.
\]
Then
\begin{equation}\label{eq:kernel}
 g(z)^\top Q^{-1}g(\zeta)=\frac{1-b(z)b(\zeta)}{1-z\zeta}.
\end{equation}
\end{lemma}

\begin{proof}
Put
\[
 M=(I-\zeta S^\top)Q(I-zS).
\]
Then \(g(z)^\top Q^{-1}g(\zeta)=w^\top M^{-1}w\). The matrix determinant lemma
and Lemma~\ref{lem:factor} give
\begin{multline*}
 1-(1-z\zeta)g(z)^\top Q^{-1}g(\zeta)\\
 =\frac{\det(M-(1-z\zeta)ww^\top)}{\det M}.
\end{multline*}
Using \eqref{eq:factor}, the quotient on the right becomes
\[
 \frac{\det(S^\top-zI)}{\det(I-zS)}
 \frac{\det(S-\zeta I)}{\det(I-\zeta S^\top)}=b(z)b(\zeta).
\]
The two factors \((-1)^n\) cancel, and
\(\det(I-\zeta S^\top)=\det(I-\zeta S)\). Rearrangement proves
\eqref{eq:kernel}.
\end{proof}

\begin{lemma}[Reciprocity]\label{lem:reciprocal}%
For \(z\ne0\), wherever both quotients are defined, the determinant quotient
in Lemma~\ref{lem:kernel} satisfies
\begin{equation}\label{eq:reciprocal}
 b(z)b(1/z)=1.
\end{equation}
\end{lemma}

\begin{proof}
Indeed,
\begin{align*}
 \det(z^{-1}I-S)&=z^{-n}\det(I-zS),\\
 \det(I-z^{-1}S)&=z^{-n}\det(zI-S),
\end{align*}
and the two determinant quotients therefore cancel.
\end{proof}

The factorization \eqref{eq:factor} contains no inverses. The additional
hypotheses in Lemma~\ref{lem:kernel} are needed only to divide by the displayed
matrices and scalars. This distinction will allow the final polynomial
identity to remain valid on \(\mathcal D\), including at singular points.

\section{Two evaluations, product reversal, and the theorem}\label{sec:sampling}

Apply Lemma~\ref{lem:kernel} with \(S=A_2\) and write
\[
 b_i=b(\alpha_i),\qquad B=\diag(b_1,\ldots,b_n).
\]
The denominator of \(b_i\) is nonzero by \(\delta_R\ne0\); its numerator is
nonzero because \(\sigma(A)\cap\sigma(A_2)=\varnothing\) on
\(\mathcal U\). Thus \(B\) is invertible.

Let \(e_i\) denote the \(i\)th standard coordinate vector. Since \(A\) is diagonal,
\eqref{eq:P} and \eqref{eq:R} give
\begin{align}
 P_{ij}&=\frac{v_iv_j}{1-\alpha_i\alpha_j}, \notag\\
 e_i^\top R&=v_iw^\top(I-\alpha_iA_2)^{-1}
          =v_i g(\alpha_i)^\top. \label{eq:modal-rows}
\end{align}
Evaluating \eqref{eq:kernel} at
\(z=\alpha_i,\zeta=\alpha_j\) therefore yields
\[
 (RQ^{-1}R^\top)_{ij}
 =\frac{v_iv_j(1-b_ib_j)}{1-\alpha_i\alpha_j}
 =P_{ij}-(BPB)_{ij},
\]
and hence
\begin{equation}\label{eq:direct}
 RQ^{-1}R^\top=P-BPB.
\end{equation}

For the evaluation at the reciprocal eigenvalue, the rank-one update alone gives
\[
 (\alpha_iI-A_2^\top)e_i=-v_iw,\qquad
 g(1/\alpha_i)=-\frac{\alpha_i}{v_i}e_i.
\]
Using Lemma~\ref{lem:reciprocal} in \eqref{eq:kernel} at
\(z=1/\alpha_i,\zeta=1/\alpha_j\) gives
\[
 b_ib_j(Q^{-1})_{ij}
 =\frac{v_iv_j(1-b_ib_j)}{1-\alpha_i\alpha_j}.
\]
Comparison with \eqref{eq:modal-rows} proves
\begin{equation}\label{eq:recipsample}
 BQ^{-1}B=P-BPB.
\end{equation}
Thus the two evaluations of one scalar kernel produce
\begin{equation}\label{eq:common-sample}
 \boxed{RQ^{-1}R^\top=BQ^{-1}B=P-BPB.}
\end{equation}

From the right equality in \eqref{eq:common-sample},
\[
 Q^{-1}=B^{-1}PB^{-1}-P,
\]
so
\begin{equation}\label{eq:target-products}
 (I+PQ)^{-1}=Q^{-1}BP^{-1}B.
\end{equation}
The left equality gives
\[
 P^{-1}RQ^{-1}R^\top=P^{-1}BQ^{-1}B.
\]
With \(M_1=P^{-1}B\) and \(M_2=Q^{-1}B\), the two target matrices are
\(M_1M_2\) and \(M_2M_1\). Therefore
\begin{equation}\label{eq:single}
 \det(\lambda I-P^{-1}RQ^{-1}R^\top)
 =\det(\lambda I-(I+PQ)^{-1})
\end{equation}
on \(\mathcal U\), by the general identity
\(\det(\lambda I-M_1M_2)=\det(\lambda I-M_2M_1)\), the classical
\(AB\)--\(BA\) principle \citep{Flanders1951}. Since \(M_1\) is invertible on
\(\mathcal U\), the two matrices are in fact similar there through
\[
 M_1^{-1}=B^{-1}P.
\]

\subsection{Proof on the natural domain}\label{sec:proof}

\begin{proof}[Proof of Theorem~\ref{thm:main}]
Equation~\eqref{eq:single} proves equality of the two characteristic
polynomials on the nonempty open set \(\mathcal U\).  On \(\Omega\), the
entries of \(P,Q,R\) and of every displayed inverse are rational functions of
the entries of \((A,v,w)\), with denominators nonzero there. Hence every
coefficient of the difference of the two characteristic polynomials is a
rational function satisfying the hypotheses of Lemma~\ref{lem:seed}.  It
vanishes throughout \(\Omega\).  Since \(\Omega\) is exactly the part of
\(\mathcal D\) where \(P,Q,I+PQ\) are nonsingular, this is the conclusion of
Theorem~\ref{thm:main}.
\end{proof}

The inverse-free continuation in
Appendix~\ref{app:polynomial-continuation} strengthens this proof: after all
denominators have been cleared, the resulting polynomial identity holds on
the whole domain \(\mathcal D\), including singular data.  It is not needed
for the preceding proof of Theorem~\ref{thm:main} on \(\Omega\).

\begin{corollary}[Determinant identity]\label{cor:det}%
For $(A,v,w)\in\mathcal D$ --- no two eigenvalues among
$\sigma(A)\cup\sigma(A_2)$ with product $1$, and no nonsingularity
or stability hypothesis ---
$\det(R)^2\,\det(I+PQ) = \det(P)\det(Q)$.
\end{corollary}

\begin{proof}
On \(\mathcal U\), set \(\lambda=0\) in \eqref{eq:single} and take
determinants.  Since the matrices there are nonsingular, this gives
\[
 \det(R)^2\det(I+PQ)=\det(P)\det(Q).
\]
Now put
\[
 \widehat P=\delta_P P,\quad
 \widehat Q=\delta_Q Q,\quad
 \widehat R=\delta_R R,\quad
 \widehat N=\delta_P\delta_Q I+\widehat P\widehat Q.
\]
By Cramer's rule these numerator matrices have polynomial entries; see
Appendix~\ref{app:clearing-details}.  Clearing the denominators in the last
determinant equality therefore shows that the polynomial
\[
 \Delta:=\det(\widehat R)^2\det(\widehat N)
   -\delta_R^{2n}\det(\widehat P)\det(\widehat Q)
\]
vanishes on \(\mathcal U\).  Lemma~\ref{lem:seed} gives
\(\Delta\equiv0\) on the entire parameter space.  Finally, on \(\mathcal D\)
the three \(\delta\)'s are nonzero.  Substituting the displayed definitions
into \(\Delta=0\) and cancelling their powers proves the stated identity,
without assuming that \(P,Q,R\), or \(I+PQ\) is nonsingular.
\end{proof}

\section{An explicit similarity in the original state coordinates}\label{sec:intertwiner}

In realization theory, system similarity is the invertible state-coordinate
relation connecting two minimal realizations of the same transfer function
\citep{ForsterNagy2004}. The result below instead constructs an explicit
ordinary matrix similarity between the two spectral matrices in
Theorem~\ref{thm:main}; the citation provides systems-theoretic background,
not a proof of the Lyapunov identity.

Let
\[
 \chi_{A_2}(z)=\det(zI-A_2),\qquad
 \chi_{A_2}^*(z)=\det(I-zA_2).
\]
Evaluating these two polynomials at \(A\), set
\begin{equation}\label{eq:bA}
 b(A)=\chi_{A_2}(A)\,\chi_{A_2}^*(A)^{-1}.
\end{equation}
The second factor is invertible on \(\mathcal D\), since
\[
 \det\chi_{A_2}^*(A)=\prod_{i,l}(1-\alpha_i\beta_l)=\delta_R.
\]
The first is invertible precisely when
\(\operatorname{res}(\chi_A,\chi_{A_2})\ne0\).

\begin{proposition}[Extended Lyapunov-kernel similarity]\label{prop:intertwiner}%
On \(\Omega^\circ\), the matrix
\begin{equation}\label{eq:generic-similarity}
 Z=b(A)^{-1}P
\end{equation}
is invertible and satisfies
\[
 Z\bigl(P^{-1}RQ^{-1}R^\top\bigr)Z^{-1}=(I+PQ)^{-1}.
\]
\end{proposition}

\begin{proof}
Set
\[
 M=P^{-1}RQ^{-1}R^\top,\qquad N=(I+PQ)^{-1}.
\]
In modal coordinates on \(\mathcal U\), \(b(A)=B\), and hence
\(Z=B^{-1}P=M_1^{-1}\).  The product reversal therefore gives
\begin{equation}\label{eq:similarity-relation}
 ZM=NZ.
\end{equation}
Under a state-coordinate change,
\[
 \begin{aligned}
  \widetilde A&=T^{-1}AT,&
  b(\widetilde A)&=T^{-1}b(A)T,\\
  \widetilde P&=T^{-1}PT^{-\top}.&&
 \end{aligned}
\]
and therefore
\[
 \widetilde Z=b(\widetilde A)^{-1}\widetilde P
             =T^{-1}ZT^{-\top}.
\]
Together with Lemma~\ref{lem:transport}, these laws make
\eqref{eq:similarity-relation} independent of the chosen coordinates.
After clearing the displayed inverses, every entry of
\[
 Z(P^{-1}RQ^{-1}R^\top)-(I+PQ)^{-1}Z
\]
has a polynomial numerator and a denominator nonzero on \(\Omega^\circ\).
Its numerator vanishes on \(\mathcal U\), hence vanishes identically by
Lemma~\ref{lem:seed}.  Thus \eqref{eq:similarity-relation} holds throughout
\(\Omega^\circ\).  There \(P\), \(\chi_{A_2}(A)\), and
\(\chi_{A_2}^*(A)\) are invertible, so \(Z\) is invertible and the
matrix relation is the stated similarity.
\end{proof}

Repeated or defective eigenvalues, and singular \(A\) or \(A_2\), do not
obstruct this formula. What it additionally requires is
\(\sigma(A)\cap\sigma(A_2)=\varnothing\). The characteristic-polynomial
identity of Theorem~\ref{thm:main} remains valid when \(A\) and \(A_2\) have a
common eigenvalue; no nonsingular limiting similarity matrix is asserted at
such a point.

\section{Domain and rank-one scope}\label{sec:comparison}

For Theorem~\ref{thm:main} it is enough to clear the rational
characteristic-polynomial difference on the natural nonsingular domain
$\Omega$. Clearing all Lyapunov-operator and matrix-inverse denominators proves
more: it constructs polynomial numerator matrices on the entire parameter
space, proves \eqref{eq:clearedhom} coefficient by coefficient, and then
evaluates it on all of $\mathcal D$. This is why Corollary~\ref{cor:det}
remains valid at
points where $P,Q$, or $R$ is singular.

Conversely, the explicit similarity matrix has a smaller natural domain. The
denominator $\chi_{A_2}^*(A)$ is controlled by $\delta_R$, but the
numerator $\chi_{A_2}(A)$ is invertible only when \(A\) and \(A_2\) have
disjoint spectra. Thus $Z=b(A)^{-1}P$ extends from the open set \(\mathcal U\)
through $\Omega^\circ$, whereas the characteristic-polynomial identity remains
valid when \(A\) and \(A_2\) have a common eigenvalue. Equality of spectra alone
does not justify extending $Z$ through a point at which it becomes
singular.

The rank-one assumption enters twice and nowhere else: the forcing term in
\eqref{eq:factor} is $ww^\top$, so the determinant lemma produces a
scalar quotient, and evaluation at the reciprocal eigenvalue uses
$(\alpha_iI-A_2^\top)e_i=-v_iw$. A genuinely higher-rank update would
replace both by matrix-valued relations. It cannot be obtained by merely
changing $b$ into a block diagonal matrix; this agrees with the known
higher-rank counterexamples recorded with Problem~9.1.

\section{Consequences}\label{sec:consequences}

This section collects what the proof makes explicit. The principal-angle
spectra of Corollary~\ref{cor:angles} are not new --- they are the
subspace canonical correlations studied by De~Cock~\citep{DeCock2002,%
DeCockDeMoor2002}; the present contribution is that the proof renders each
as the spectrum of an explicit matrix. The similarity matrix $Z$ of
Proposition~\ref{prop:intertwiner} is specific to this argument.

\begin{corollary}[Principal-angle family]\label{cor:angles}
In the Schur-stable case, assume in addition that $(A,v)$ is
controllable and $(A_2,w)$ observable, so that the Gramians
$P=\mathcal C_\infty\mathcal C_\infty^\top$,
$Q=\mathcal O_\infty\mathcal O_\infty^\top$ (and
$R=\mathcal C_\infty\mathcal O_\infty^\top$) have factors of full row rank
and $P,Q\succ0$. Then the principal angles $\theta_i$
between $\range(\mathcal C_\infty^\top)$ and
$\range(\mathcal O_\infty^\top)$~\citep{BjorckGolub1973}
have all three squared-angle spectra as explicit matrix spectra.
\[
  \{\cos^2\theta_i\}_{i=1}^{n}=\sigma\bigl((I+PQ)^{-1}\bigr).
\]
\[
  \{\tan^2\theta_i\}_{i=1}^{n}=\sigma(PQ).
\]
\[
  \{\sin^2\theta_i\}_{i=1}^{n}=\sigma\bigl(PQ(I+PQ)^{-1}\bigr).
\]
The $\cos^2$ form is Theorem~\ref{thm:main}. The $\cos\theta_i$ are the
\emph{subspace} canonical correlations, the $r_i$ the \emph{past/future
process} ones (Remark~\ref{rem:conventions}); the two notions are related by
the Pythagoras complement $\sin^2=1-\cos^2$, with $\tan^2=\sin^2/\cos^2$.
\end{corollary}

\begin{remark}[Hankel singular values]\label{rem:hankel}
Classically, the Hankel singular values associated with Gramians $P,Q$ are
the nonnegative square roots of the eigenvalues in $\sigma(PQ)$, counted with
algebraic multiplicity \citep{LindquistPicci1996}. Thus
Corollary~\ref{cor:angles} identifies them, in this paired-Gramian setting,
with the tangents $\tan\theta_i$. Here $P$ and $Q$ are Gramians of the two
\emph{different} state matrices $A$ and $A_2$, so the statement concerns the
Gramian pair, not the Hankel values of a single realisation.
\end{remark}

\section{Remarks}\label{sec:remarks}

\begin{remark}[Cepstral norm]\label{rem:cepstral}
In the Schur-stable stochastic interpretation, Corollary~\ref{cor:det} becomes
the usual cepstral/mutual-information identity
\[
  \log\det(I+PQ)=-\sum_i\log\bigl(1-r_i^{2}\bigr),
\]
where $r_i$ are the past/future canonical correlations; the left
side is the squared Martin cepstral norm \citep{Martin2000}, equal for a
Gaussian process to twice the mutual information $I(Y_{\mathrm p};Y_{\mathrm f})$,
where $Y_{\mathrm p}:=(\ldots,y(-2),y(-1))$ and
$Y_{\mathrm f}:=(y(0),y(1),\ldots)$ denote the past and future output sequences
(\citealp{GelfandYaglom1959}; in the realization setting,
\citealp{Akaike1975}; \citealp{JewellBloomfieldBartmann1983}).
\end{remark}

\begin{remark}[Relation to previous work]\label{rem:prevwork}
The realization-theoretic treatments \citep{DeCock2002,DeCockDeMoor2002}
establish the equivalent canonical-correlation relation on the Schur-stable
stochastic face, and \citet{Scherrer2002} treats the associated relation
there through minimum-phase balanced truncation and operator methods, using
the Lindquist--Picci canonical-correlation identity \citep{LindquistPicci1996}.
More generally, the connection between all-pass transfer matrices and
future/past canonical correlations is classical; see \citet{Green1988} and
\citet{LindquistPicci1996}.
The exact matrix identity was recorded as Conjecture~3.13 of the thesis
\citep{DeCock2002} and, in 2004, as Problem~9.1. The present contribution has
a different form and scope: a direct finite-dimensional Lyapunov-kernel
similarity on a nonempty full-dimensional open set and a polynomial extension
to the full domain excluding reciprocal eigenvalue products, with no stability
or definiteness
hypothesis. We therefore describe it as a direct matrix-algebraic proof and
full-domain extension of the statement of Problem~9.1, rather than as a first
proof on any open set.

Evaluating the single kernel \eqref{eq:kernel} at the eigenvalues and at their
reciprocals leads to the reversal of the factors \(M_1,M_2\). The
denominator-cleared identity
\eqref{eq:clearedhom} extends through defective and repeated matrices, through
cases where \(A\) and \(A_2\) have a common eigenvalue, and through singular-node
data in $\mathcal D$, and yields
Corollary~\ref{cor:det} even where $P$ or $Q$ is singular. The explicit
similarity matrix $Z=b(A)^{-1}P$ has the separate domain $\Omega^\circ$
(Proposition~\ref{prop:intertwiner}).
\end{remark}

\section{Conclusion}\label{sec:conclusion}

We have given a direct matrix-algebraic proof of the rank-one Lyapunov
spectral identity of Problem~9.1. One determinant factorization of the
Lyapunov equation produces a bilinear kernel; evaluating it at the eigenvalues
and at their reciprocals gives
$RQ^{-1}R^\top=BQ^{-1}B=P-BPB$, and the target matrices become
the products $M_1M_2$ and $M_2M_1$. On $\Omega^\circ$ the argument supplies the
explicit similarity matrix $Z=b(A)^{-1}P$ in the original state coordinates.
After the
denominators are cleared, polynomial uniqueness carries the
characteristic-polynomial identity to the full domain excluding reciprocal
eigenvalue products and yields the determinant corollary without
nonsingularity assumptions. The
identity gives the principal-angle and canonical-correlation spectra as
explicit matrix spectra (Section~\ref{sec:consequences}).

When $A$ and $A_2$ have a common eigenvalue, the spectral identity remains
valid, but the displayed similarity matrix need not remain nonsingular; the
similarity statement is therefore restricted to disjoint spectra. The
unrestricted multicolumn extension is false, although a
simultaneous block reduction to independent rank-one channels inherits the
identity blockwise. Beyond such reducible cases, it remains open whether
genuinely coupled higher-rank data admit an analogue under an appropriate
matrix-valued Lyapunov-kernel or all-pass hypothesis.

\appendix

\section{Construction of the open set}\label{app:generic-locus}

This appendix supplies the explicit witness calculation used in
Lemma~\ref{lem:seed}.  Choose
\[
 \begin{gathered}
  \alpha_i^0=\frac{i}{2(n+1)},\qquad
  A_0=\diag(\alpha_1^0,\ldots,\alpha_n^0),\\
  v_0=w_0=\tau\one,\qquad t=\tau^2,
 \end{gathered}
\]
with $t>0$ sufficiently small that
\[
             t\sum_{i=1}^n\frac{1}{1-\alpha_i^0}<1.
\]
For $\lambda\ne\alpha_i^0$, the rank-one determinant formula gives
\[
 \det(\lambda I-A_0-t\one\one^\top)
 =\prod_i(\lambda-\alpha_i^0)
   \left(1-t\sum_i\frac{1}{\lambda-\alpha_i^0}\right).
\]
The parenthesized function is strictly increasing between consecutive poles,
with limits $-\infty$ and $+\infty$ on every finite gap and limit $1$ at
$+\infty$.  It has one simple zero in every
$(\alpha_i^0,\alpha_{i+1}^0)$ and one above $\alpha_n^0$.  The smallness
condition makes its value at $1$ positive, so the last zero is also below $1$.
Thus the eigenvalues of $A_{2,0}=A_0+t\one\one^\top$ strictly interlace those
of $A_0$; both matrices are Schur stable, their spectra are simple and
disjoint, and every product drawn from the two spectra differs from $1$.

The controllability matrix of $(A_0,v_0)$ is a row-scaled Vandermonde matrix,
with determinant
\[
 \tau^n\prod_{i<j}(\alpha_j^0-\alpha_i^0)\ne0.
\]
The symmetric matrix $A_{2,0}$ is observable from $w_0^\top$: an eigenvector
$\xi\ne0$ satisfying $\one^\top\xi=0$ would give
$A_{2,0}\xi=A_0\xi$, contradicting
spectral disjointness.  The convergent Gramian series therefore give
$P_0,Q_0\succ0$.  Moreover, $P_0Q_0$ is similar to
$P_0^{1/2}Q_0P_0^{1/2}\succ0$, so $I+P_0Q_0$ is nonsingular.  The witness lies
in $\Omega^\circ$.

Take $\mathcal U$ to be a sufficiently small common nonvanishing neighborhood
of this witness for
\begin{gather*}
 \delta_P,\ \delta_Q,\ \delta_R,\qquad
 \disc\chi_A,\qquad \det A,\\
 \det[v,Av,\ldots,A^{n-1}v],\qquad
 \operatorname{res}(\chi_A,\chi_{A_2}),\\
 \det P,\qquad \det Q,\qquad \det(I+PQ).
\end{gather*}
The first seven expressions are polynomials in $(A,v,w)$.  On $\mathcal D$,
the last three are rational functions whose denominators are powers of the
nonzero $\delta$'s.  All are nonzero at the witness, so the set on which they
are simultaneously nonzero contains an open ball.  No simple spectrum or
nonsingularity condition for $A_2$, no observability condition, and no
condition on $\det A_2$ or $\det R$ is imposed on the neighborhood;
observability was used only to certify $Q_0\succ0$ at the witness.

For the identity principle, choose a common denominator for the finitely many
rational coefficients.  It is nonzero on $\Omega$.  Every resulting numerator
is a real polynomial that vanishes on the nonempty open set $\mathcal U$ and
therefore is the zero polynomial.  Dividing by the common denominator gives
the last assertion of Lemma~\ref{lem:seed}.

\paragraph{Denominator inventory for the modal calculation.}
At a point of \(\mathcal U\), \(\alpha_i\ne0\) follows from \(\det A\ne0\),
and \(v_i\ne0\) follows from controllability.  The conditions
\(1-\alpha_i\alpha_j\ne0\) and \(1-\alpha_i\beta_l\ne0\) follow from
\(\delta_P\delta_R\ne0\), while
\(\alpha_i-\beta_l\ne0\) follows from
\(\sigma(A)\cap\sigma(A_2)=\varnothing\).  Hence
\(I-\alpha_iA_2\), \(I-\alpha_iA_2^\top\),
\(\alpha_iI-A_2^\top\), and
\(I-\alpha_i^{-1}A_2^\top\) are invertible wherever they occur.
Membership in \(\Omega\) supplies nonsingular \(P,Q,I+PQ\).  These are all
the divisions used in the two kernel evaluations; no inverse of \(R\) is
introduced.

\section{Inverse-free polynomial continuation}\label{app:polynomial-continuation}

The proof in Section~\ref{sec:proof} uses only rational continuation on the
natural nonsingular domain \(\Omega\).  This appendix records the stronger
inverse-free statement on all of \(\mathcal D\), including data for which one
of the matrices occurring in the spectral formula is singular.

\begin{lemma}[Regularity on $\mathcal D$ and clearing]\label{lem:regular}%
Recall the Lyapunov-operator determinants $\delta_P,\delta_Q,\delta_R$
(Table~\ref{tab:notation-core}), and define the \emph{numerator matrices}
\[
  \widehat P:=\delta_P P,\quad \widehat Q:=\delta_Q Q,\quad \widehat R:=\delta_R R.
\]
Then \textup{(i)}~each of $\widehat P,\widehat Q,\widehat R$ is an $n\times n$
matrix of \emph{polynomials} in $(A,v,w)$ on the entire parameter space, and
$\delta_P,\delta_Q,\delta_R$ are nonzero on $\mathcal D$, so every determinant
formed from $P,Q,R$ is regular on $\mathcal D$, with denominator read directly
off the numerator matrices:
$\det P=\det\widehat P/\delta_P^{\,n}$, $\det Q=\det\widehat Q/\delta_Q^{\,n}$,
$\det R=\det\widehat R/\delta_R^{\,n}$, $\adj Q=\delta_Q^{-(n-1)}\adj\widehat Q$,
and $\det(I+PQ)=\det(\delta_P\delta_Q I+\widehat P\widehat Q)/(\delta_P\delta_Q)^{n}$.
\textup{(ii)}~The denominator-cleared relation
\begin{equation}\label{eq:cleared}
\begin{aligned}
  &\det(I+PQ)\,\det\!\bigl(\det(Q)\lambda P-R\adj(Q)R^\top\bigr)\\
  &\qquad=\det(P)\det(Q)^n\det\!\bigl((\lambda-1)I+\lambda PQ\bigr)
\end{aligned}
\end{equation}
--- the \emph{cleared equation}, whose validity on $\mathcal D$ is proved in
Lemma~\ref{lem:divideback}.
Multiplied through by the factor
\[
 (\delta_P\delta_Q)^{n}
 (\delta_P\delta_Q^{\,n}\delta_R^{2})^{n},
\]
it
becomes the identity of \emph{polynomial} matrices and determinants
\begin{multline}\label{eq:clearedhom}
\det\widehat N\,
\det\!\bigl(\delta_R^{2}\det(\widehat Q)\,\lambda\,\widehat P
      -\delta_P\delta_Q\,\widehat R\,\adj(\widehat Q)\,\widehat R^{\top}\bigr)\\
 =\delta_R^{\,2n}\,\det\widehat P\,\det(\widehat Q)^{n}
  \det\!\bigl((\lambda-1)\delta_P\delta_Q I+\lambda\,\widehat P\widehat Q\bigr),
\end{multline}
with $\widehat N:=\delta_P\delta_Q I+\widehat P\widehat Q$: every symbol in
\eqref{eq:clearedhom} is a polynomial in the entries of $(A,v,w)$, each
$\lambda^k$-coefficient of the difference of its two sides is one polynomial
in $(A,v,w)$, and on $\mathcal D$, where the $\delta$'s are nonzero,
\eqref{eq:clearedhom} holds if and only if \eqref{eq:cleared} does.
\end{lemma}

\begin{proof}
Vectorizing the three equations and applying Cramer's rule writes every entry
of $P,Q,R$ as a polynomial numerator divided respectively by
$\delta_P,\delta_Q,\delta_R$.  The determinant formulas in \textup{(i)} then
follow by multilinearity.  Substituting
$P=\widehat P/\delta_P$, $Q=\widehat Q/\delta_Q$ and
$R=\widehat R/\delta_R$ into \eqref{eq:cleared}, and collecting scalar factors,
gives \eqref{eq:clearedhom}; the converse follows by dividing by the nonzero
$\delta$'s on $\mathcal D$.  The explicit vectorized formulas and denominator
bookkeeping are recorded in Appendix~\ref{app:clearing-details}.
\end{proof}

The passage from $\mathcal U$ to all of $\mathcal D$ now follows from
determinant algebra and polynomial uniqueness.

\subsection{The continuation}\label{subsec:continuation}

\begin{lemma}[Polynomial identity and recovery of the original formula]\label{lem:divideback}%
With the denominator-cleared equation \eqref{eq:cleared} of
Lemma~\ref{lem:regular}, the following hold.
\begin{enumerate}
\item[\textup{(i)}] \emph{Equivalence when the required matrices are
nonsingular.} At any
  $(A,v,w)$ with $P,Q,I+PQ$ nonsingular, \eqref{eq:cleared} is equivalent to the
  characteristic-polynomial equality
  \[
    \det(\lambda I-P^{-1}RQ^{-1}R^\top)
      =\det(\lambda I-(I+PQ)^{-1}),
  \]
  and hence to
  $\sigma(P^{-1}RQ^{-1}R^\top)=\sigma((I+PQ)^{-1})$.
\item[\textup{(ii)}] \emph{Polynomial extension to $\mathcal D$.} Expand the
  difference of the two sides of the cleared form \eqref{eq:clearedhom} in
  $\lambda$: by Lemma~\ref{lem:regular} its coefficients are
  finitely many genuine polynomials $\widehat c_k\in\mathbb R[A,v,w]$ on the entire
  parameter space. Each $\widehat c_k$ vanishes on $\mathcal U$ (by \textup{(i)}
  and the product reversal \eqref{eq:single}), hence vanishes identically by
  Lemma~\ref{lem:seed};
  so \eqref{eq:cleared} holds throughout $\mathcal D$. (Lemma~\ref{lem:seed}
  is applied to the \emph{polynomials} $\widehat c_k$, not to a rational
  identity.)
\item[\textup{(iii)}] \emph{Recovery of the original formula.} Consequently,
  at every $(A,v,w)\in\mathcal D$ with $P,Q,I+PQ$ nonsingular, cancelling the
  nonzero determinant factors in \eqref{eq:cleared} and applying \textup{(i)}
  gives $\sigma(P^{-1}RQ^{-1}R^\top)=\sigma((I+PQ)^{-1})$.
\end{enumerate}
\end{lemma}

\begin{proof}
Write $M:=P^{-1}RQ^{-1}R^\top$ and $N:=(I+PQ)^{-1}$. Both implications are
determinant identities valid for arbitrary matrices, with
$Q^{-1}=\adj(Q)/\det(Q)$:
\[
  \det(\lambda I-M)
   =\frac{\det\!\bigl(\det(Q)\lambda P-R\adj(Q)R^\top\bigr)}{\det(P)\det(Q)^n},
\]
\[
  \det(\lambda I-N)=\frac{\det\!\bigl((\lambda-1)I+\lambda PQ\bigr)}{\det(I+PQ)}.
\]
Multiplying the first by $\det(P)\det(Q)^n$ and the second by $\det(I+PQ)$
turns the characteristic-polynomial equality into \eqref{eq:cleared} and back;
the cleared sides are polynomial by Lemma~\ref{lem:regular}; recovering the
original equality needs only $\det(P)\det(Q)\det(I+PQ)\neq0$, so no
diagonalizability enters.
\end{proof}

This identity extends through singular data and through points where \(A\) and
\(A_2\) have a common eigenvalue.  It recovers Theorem~\ref{thm:main} on
\(\Omega\), without a diagonalizability, simple-spectrum, or nonsingular-
\(R\) assumption.

\section{Polynomial numerator bookkeeping}\label{app:clearing-details}

The Kronecker--vectorized solutions of \eqref{eq:P}--\eqref{eq:R} are
\[
\begin{aligned}
 \vec(P)&=\delta_P^{-1}\adj(I-A\otimes A)\,\vec(vv^\top),\\
 \vec(Q)&=\delta_Q^{-1}\adj(I-A_2^\top\otimes A_2^\top)\,\vec(ww^\top),\\
 \vec(R)&=\delta_R^{-1}\adj(I-A_2^\top\otimes A)\,\vec(vw^\top).
\end{aligned}
\]
Thus every entry of $P,Q,R$ is rational in $(A,v,w)$ with respective
denominator $\delta_P,\delta_Q,\delta_R$ (possibly reduced after
cancellation), and $\widehat P,\widehat Q,\widehat R$ have polynomial entries.
The row-pair vectorization used in the formalization writes the last operator
with the two Kronecker factors in the opposite order.  Swapping these factors
is a simultaneous row-and-column permutation and therefore leaves its
determinant, $\delta_R$, unchanged.

On $\mathcal D$, the matrices $A\otimes A$,
$A_2^\top\otimes A_2^\top$, and $A_2^\top\otimes A$ have eigenvalues
$\alpha_i\alpha_j$, $\beta_i\beta_j$, and $\beta_i\alpha_j$, respectively,
none equal to $1$.  Hence the three $\delta$'s are nonzero, and the determinant
formulas of Lemma~\ref{lem:regular}\textup{(i)} follow by multilinearity.
Substituting $P=\widehat P/\delta_P$, $Q=\widehat Q/\delta_Q$, and
$R=\widehat R/\delta_R$ into \eqref{eq:cleared}, and absorbing the scalar
factors, gives \eqref{eq:clearedhom}.  Every displayed term is then polynomial;
division by the same nonzero scalars on $\mathcal D$ proves the converse.

\section*{CRediT authorship contribution statement}

\textbf{Jonas Gillberg:} Conceptualization, Methodology, Formal analysis,
Software, Validation, Visualization, Writing -- original draft, Writing --
review \& editing.
\textbf{Johan L\"ofberg:} Investigation. He contributed helpful
simplifications that streamlined the proof and provided constructive comments
on the manuscript.

\section*{Declaration of generative AI and AI-assisted technologies in
the writing and research process}

During the preparation of this work, the authors used ChatGPT and Codex
(OpenAI) and Claude (Anthropic) for exploratory proof development, language
editing, software development, and assistance with the Lean formalization.
The authors reviewed and edited the resulting material, checked the
mathematical arguments and computations reported in the paper, and take full
responsibility for the content of the publication.

\section*{Declaration of competing interest}

The authors declare that they have no known competing financial
interests or personal relationships that could have influenced the work
reported in this paper.

\section*{Data availability}

A Lean~4 formalization, exact computer-algebra calculations, numerical tests,
and reproduction instructions are available in the accompanying repository
snapshot. These materials support, but do not replace, the complete and
self-contained proofs in this paper. A versioned public archive with a citable
DOI will be deposited after author and independent review.

\ifarxivhtml
\section*{Acknowledgements}

The authors thank B.~De~Moor and K.~De~Cock for surfacing the problem that
motivated this note, and B.~Wahlberg for early encouragement.

\section*{Author biographies}

\subsection*{Jonas Gillberg}

Jonas Gillberg received the Ph.D.\ degree in Automatic Control from
Link\"oping University, Link\"oping, Sweden, in 2006, where his thesis work
on frequency-domain system identification was supervised by Prof.~Lennart
Ljung. He was previously with IBM~Research in the Quantum Computing division
and is currently with ZyQE in Stockholm, Sweden. His research interests in
control theory are numerical methods, system identification, and convex
optimization.

\subsection*{Johan L\"ofberg}

Johan L\"ofberg received the Ph.D.\ degree in Automatic Control from
Link\"oping University, Link\"oping, Sweden, in 2003, with a thesis on minimax
approaches to robust model predictive control. After three years as a
postdoctoral researcher at the Automatic Control Laboratory, ETH Z\"urich,
Switzerland, he returned to Link\"oping University, where he was appointed
Docent in 2011 and is currently Associate Professor in the Division of
Automatic Control, Department of Electrical Engineering. His research lies at
the border between control theory and optimization, and he is the developer
of the YALMIP modelling language for optimization in MATLAB.
\else
\ack{The authors thank B.~De~Moor and K.~De~Cock for surfacing the
problem that motivated this note, and B.~Wahlberg for early
encouragement.}

\biographyentry{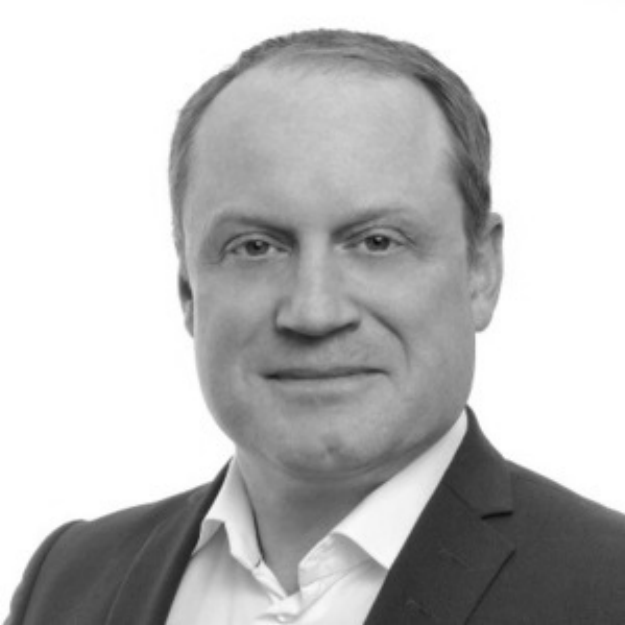}{Jonas Gillberg}{%
received the Ph.D.\ degree in Automatic Control from Link\"oping
University, Link\"oping, Sweden, in 2006, where his thesis work on
frequency-domain system identification was supervised by
Prof.~Lennart Ljung.  He was previously with IBM~Research in the
Quantum Computing division and is currently with ZyQE in
Stockholm, Sweden.  His research interests in
control theory are numerical methods, system identification, and convex
optimization.}

\biographyentry{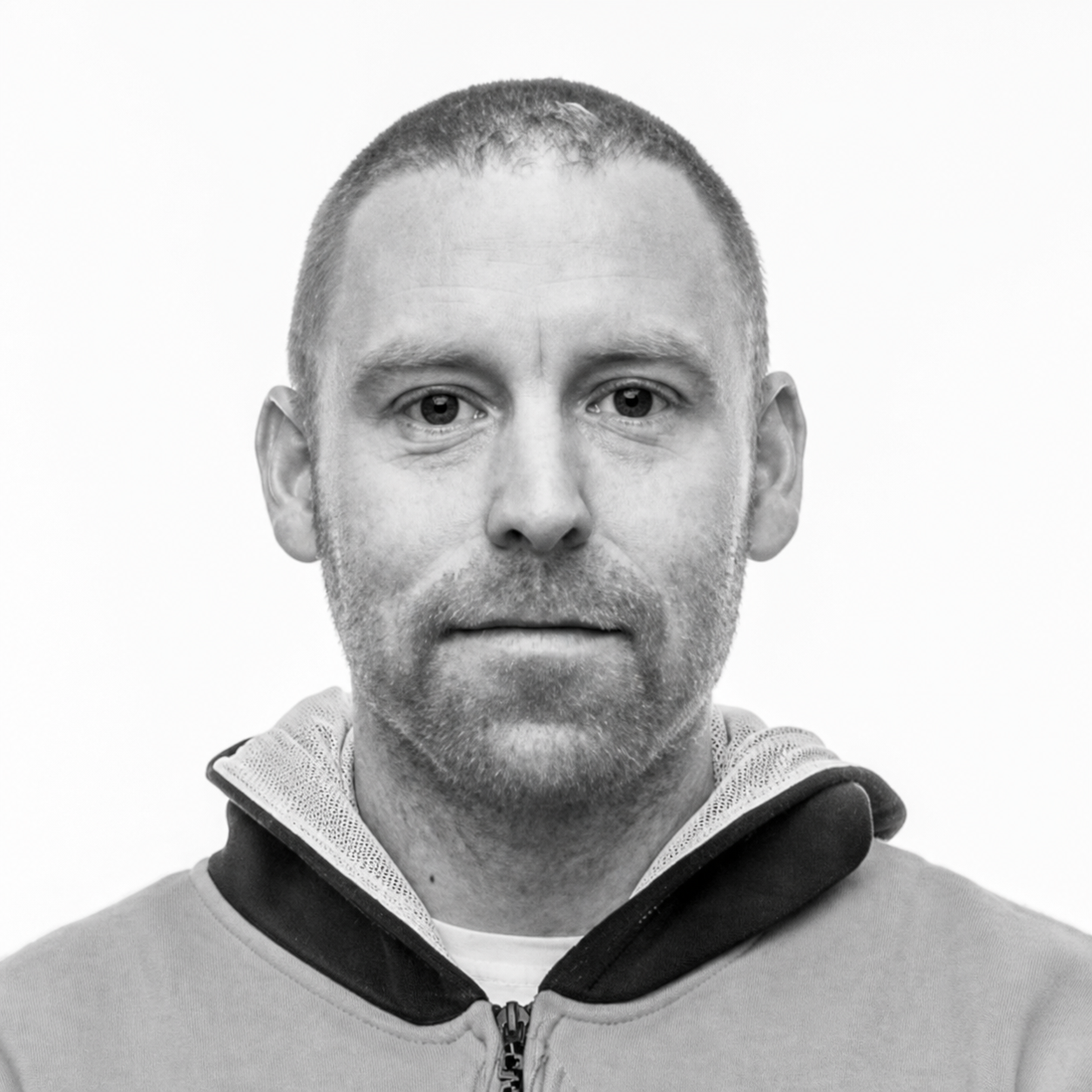}{Johan L\"ofberg}{%
received the Ph.D.\ degree in Automatic Control from Link\"oping
University, Link\"oping, Sweden, in 2003, with a thesis on minimax
approaches to robust model predictive control.  After three years as a
postdoctoral researcher at the Automatic Control Laboratory, ETH
Z\"urich, Switzerland, he returned to Link\"oping University, where he
was appointed Docent in 2011 and is currently Associate Professor in the
Division of Automatic Control, Department of Electrical Engineering.  His
research lies at the border between control theory and optimization, and
he is the developer of the YALMIP modelling language for optimization
in MATLAB.}
\fi

\begingroup\small
\ifarxivhtml

\else
\bibliographystyle{elsarticle-harv}
\bibliography{paper}
\fi
\endgroup

\end{document}